\documentclass[11pt]{article}
\usepackage{fullpage}
\usepackage{amsmath}
\usepackage{cite}
\usepackage{amsmath}
\usepackage{amsthm}
\usepackage{xcolor}
\usepackage{amsfonts}
\usepackage{amssymb}
\usepackage{caption}
\usepackage{mathtools}
\usepackage{relsize}
\usepackage{algpseudocode}
\usepackage{algorithm}
\usepackage{graphicx}
\usepackage{bbm}
\usepackage{mathrsfs}
\usepackage{hyperref}
\newcommand{\norm}[1]{\left\lVert#1\right\rVert}
\DeclarePairedDelimiter\ceil{\lceil}{\rceil}
\DeclarePairedDelimiter\floor{\lfloor}{\rfloor}
\algnewcommand{\IIf}[1]{\State\algorithmicif\ #1\ \algorithmicthen}
\algnewcommand{\EndIIf}{\unskip\ \algorithmicend\ \algorithmicif}

\def\bool{\{0,1\}}
\newcommand{\ket}[1]{|#1\rangle}
\newcommand{\bra}[1]{\langle #1|}
\newcommand{\ignore}[1]{}

\newtheorem{theorem}{Theorem} 
\newtheorem*{theorem*}{Theorem}
\newtheorem{lemma}[theorem]{Lemma}
\newtheorem*{lemma*}{Lemma}

\newtheorem*{proposition*}{Proposition}
\newtheorem{definition}{Definition} 
\title{{\bf A Quantum ``Lifting Theorem'' for Constructions of Pseudorandom Generators from Random Oracles}}
\author{Jonathan Katz\thanks{Dept.\ of Computer Science, University of Maryland. {\bf Email:} {\tt jkatz2@gmail.com}, {\tt benjsela@cs.umd.edu}. Work supported by NSF award CNS-2154705.} \and Benjamin Sela$^*$}
\date{}

\def\A{\mathcal{A}}
\def\B{\mathcal{B}}

\newcommand{\qeq}[1]{=}

\newcommand{\extFunc}[1]{{\mathbbm{1}^{(#1)}}}

\begin{document}

\maketitle

\begin{abstract} We study the (quantum) security of pseudorandom generators (PRGs) constructed from random oracles. 
We prove a ``lifting theorem'' showing, roughly, that if such a PRG is unconditionally secure against classical adversaries making polynomially many queries to the random oracle, then it is also (unconditionally) secure against quantum adversaries in the same sense. 
As a result of independent interest, we also show that any pseudo-deterministic quantum-oracle algorithm (i.e., a quantum algorithm that with high probability returns the same value on repeated executions) can be simulated by a computationally unbounded---but query bounded---classical-oracle algorithm with only a polynomial blowup in the number of queries. This implies as a corollary that our lifting theorem holds even for PRGs that themselves make quantum queries to the random oracle.
\end{abstract}

\section{Introduction}

The random-oracle model (ROM) has long been a useful heuristic for proving the security of cryptosystems that make black-box use of a cryptographic hash function \cite{BR93}. Boneh et al.~\cite{boneh_random_2011} 
observed the need to consider the quantum random-oracle model (QROM) when dealing with quantum adversaries, to take into account the fact that such attackers can evaluate the random oracle on a superposition of inputs.
Although many security results in the ROM can be translated to the~QROM~\cite{boneh_random_2011,yamakawa_classical_2021,zhang_quantum_2019,kramer_encryption_2020}, this is not always the case.
Indeed, Zhang et al.~\cite{zhang_quantum_2019} show that, in general, an adversary making polynomially many quantum queries to a random oracle cannot be simulated by an adversary making polynomially many classical queries to the random oracle. 
Yamakawa and Zhandry extended this result to show 
cryptosystems that are secure in the ROM but not in the QROM, first for the case of digital signatures and public-key encryption schemes~\cite{yamakawa_classical_2021}, and then for one-way functions and collision-resistant hash functions~\cite{yamakawa_verifiable_2022}. 

The above line of work left open the question of showing separations or translation results for other cryptographic primitives constructed from random oracles.
For the particular case of pseudorandom generators (PRGs), Yamakawa and Zhandry suggest (without proof) that if the Aaronson-Ambainis conjecture~\cite{aaronson_need_2014} holds\footnote{The Aaronson-Ambainis conjecture roughly states that any quantum-oracle algorithm outputting a single bit can be simulated by a classical-oracle algorithm making a similar number of queries.} then any unconditionally secure construction of a PRG in the ROM remains secure in the~QROM.
Here, we prove this claim by Yamakawa and Zhandry; in fact, we show the claim is true  
without assuming the Aaronson-Ambianis conjecture. 
Specifically, we prove a ``lifting theorem'' showing that for any deterministic algorithm $G$ having access to a random oracle, if there is an attacker distinguishing the output of $G$ from random using $Q$ \emph{quantum} queries to the random oracle then there is an attacker distinguishing the output of $G$ from random using ${\sf poly}(Q)$ \emph{classical} queries to the random oracle.
As a result of additional interest, we analyze the behavior of so-called \emph{pseudo-deterministic} algorithms, i.e., quantum-oracle algorithms that, with high probability over the measurement randomness of their final state, output the same value on repeated executions (for any oracle). We prove that any pseudo-deterministic algorithm making $Q$ quantum queries to an oracle can be simulated by an algorithm making ${\sf poly}(Q)$ classical queries to the same oracle. An easy corollary is that our aforementioned result holds even if $G$ makes quantum queries to its random oracle.

\subsection{Overview of our Techniques} 
We give a high-level overview of our results.

\medskip\noindent{\bf PRG lifting theorem.} Consider a deterministic (classical) algorithm $G$ based on a random oracle~$H$.
We want to show that any algorithm $\A_{\sf quantum}$ distinguishing the output of $G$ (on a uniform seed) from a random string, and making $Q$ quantum queries to~$H$ can be emulated by an algorithm $\A_{\sf classical}$ making ${\sf poly}(Q)$ classical queries to~$H$.
Intuitively, our proof relies on the fact that the distinguishing advantage of $\A_{\sf quantum}$ can only come from  correlations between the output of $G$ and the random oracle.  
Since $G$ only makes polynomially many queries to~$H$, there are only polynomially many oracle inputs that are ``usefully'' correlated with the output, and those queries can be made by~$\A_{\sf classical}$.

More concretely, say an oracle input is useful if it has some nonnegligible probability of being queried by~$G$, where the probability is over the choice of the seed and the random oracle. There can only be polynomially many such inputs.
Our classical distinguisher $\A_{\sf classical}$ proceeds by first querying the oracle $H$ on those ``useful'' points, and then extending its partial view of the random oracle to a function $H'$ defined on the entire domain by uniformly generating the remainder of the function values. We prove that the distinguishing advantage of $\A_{\sf quantum}$ given oracle $H'$ is almost as large as its advantage when given access to the original oracle~$H$. Since $\A_{\sf classical}$ knows the entire function~$H'$, it can simulate $\A_{\sf quantum}$ without making any more queries to~$H$ and the result follows.

\medskip\noindent{\bf Pseudo-deterministic algorithms.} The above considers a classical PRG, but we would like to extend the results to the case where $G$ itself makes quantum queries to~$H$. In that case, it seems natural to require that $G$ satisfies the following property: for any $H$ and any seed~$s$, there is a value $g$ such that $G^{\ket{H}}(s)$ outputs $g$ with overwhelming probability over its measurement randomness. A quantum algorithm with that property is called pseudo-deterministic \cite{bartusek_obfuscation_2023}. We show that any pseudo-deterministic oracle algorithm can be simulated by a computationally unbounded classical algorithm making a polynomialy related number of queries to the same oracle. This allows us to replace the quantum PRG with a classical one, at which point we can apply our previous result.

As intuition for our proof, consider a pseudo-deterministic algorithm $\A_{\sf quantum}$ interacting with an oracle~$H$. We show that even though the answer to a quantum query to $H$ might depend on the value of the oracle at every point in its domain, the output of $\A_{\sf quantum}$ can only depend on the value of~$H$ at a polynomial number of ``critical points.'' (Yamakawa and Zhandry~\cite{yamakawa_verifiable_2022} mention this fact in passing without proof.) We not only prove this statement (cf.\ Lemma~\ref{lemma:oracleDependance}), but also provide a way to identify those critical points. We then present an algorithm which, given any pseudo-deterministic oracle algorithm, makes a polynomially related number of classical queries to the same oracle and outputs the same value as the quantum algorithm with high probability.

\subsection{Related Work}

There has been much work on translating proofs of security in the ROM to the QROM, generally by showing that if a security proof in the ROM satisfies certain conditions on how it interacts with the random oracle, then the security lifts to the setting where an adversary has quantum access to the random oracle for free. For instance, in their work introducing the QROM~\cite{boneh_random_2011}, Boneh et al.\ defined  \emph{history-free} reductions and showed that security proofs satisfying this definition can be lifted from the ROM to the QROM. Song~\cite{song_note_2014} introduced a more detailed framework for translating certain classes of reductions from the ROM to the QROM, and as a consequence proved a lifting theorem for the Full Domain Hash (FDH) signature scheme. Zhang et al.~\cite{zhang_quantum_2019} formalized \emph{committed programming reductions} in which the strategy by which the simulator reprograms the random oracle is required to be fixed prior to interacting with the adversary. They not only show that this restricted class of reductions can be used to prove security of a handful of well-known cryptographic schemes, but that such reductions can be lifted to the QROM. Yamakawa and Zhandry~\cite{yamakawa_classical_2021} give a lifting theorem for something they call a search-type game as well as for digital signature schemes that are restricted in how they make use of the ROM; that model captures the FDH signature scheme as well as signatures constructed using the Fiat-Shamir transform. Kramer and Struck~\cite{kramer_encryption_2020} studied the class of so-called ``oracle-simple'' PKE schemes and show that security of such schemes in the ROM lifts to the~QROM.

We note that the above works largely either focus on translating classes of security proofs from the ROM to the QROM, or showing that security can be lifted if the cryptographic scheme itself satisfies certain restrictions on how it interacts with the random oracle. In contrast, we wish to give a lifting theorem for the security for an arbitrary PRG construction where we cannot argue that a certain security proof can be employed in the ROM, and where none of the aforementioned random oracle restrictions can be assumed.

\subsection{Open Questions} 
We conjecture that our result can be extended to constructions of pseudorandom functions (PRFs) in the ROM. 
Note that there are PRFs that are secure for distinguishers making classical queries to the PRF but not for distinguishers making quantum queries to the~PRF \cite{6375347}. Thus, any lifting theorem in this setting should consider distinguishers making only classical queries to the PRF but making classical/quantum queries to the random oracle.
A starting point might be to try to extend our results to the case of \emph{weak} PRFs, where the distinguisher is not given the ability to make any queries to the PRF at all.

\section{Preliminaries}
For a finite set $X$, we write $x\leftarrow X$ to denote that $x$ is sampled uniformly from $X$. For a distribution~$\mathcal{X}$, we write $x\leftarrow \mathcal{X}$ to denote that $x$ is sampled according to~$\mathcal{X}$. We let $\mathbbm{1}$ be the identity function, where the domain and range will be clear from context. For a function $H$, we let $D_{H}$ denote the domain of~$H$. 
Let ${\sf Func}_{n,m}=\{H:\{0,1\}^{n}\mapsto \{0,1\}^{m}\}$. If $H:D_H \mapsto \bool^m$, with $D_H \subseteq \bool^n$, is a (partial) function, and 
$h$ is a partial function whose domain is a subset of $D_H$, we say that $H$ and $h$ are \emph{consistent} if $H(x)=h(x)$ for all $x \in D_{h}$; we write ${\sf Func}_{n,m}(h)$ to denote the subset of ${\sf Func}_{n,m}$ that is consistent with~$h$. 
For a partial function $f$ (not necessarily consistent with~$H$) such that $D_{f}\subseteq D_{H}$,  we write $H^{(f)}$  to denote the following function:
\[ H^{(f)}(x) = \begin{cases} 
      f(x) & x \in D_{f} \\
      H(x) & \text{otherwise}. 
   \end{cases}
\]
For a partial function $f$ and a function $H \in {\sf Func}_{n,m}(f)$, we let $H\setminus f$ denote the partial function consistent with $H$ and defined exactly where $f$ is not defined. Similarly, for two partial functions $f,g$ which do not disagree on any value for which they are both defined, we let $f\cup g$ denote the (partial) function with domain $D_{f}\cup D_{g}$ that is consistent with both functions. We let ${\sf negl}(\cdot)$ denote an unspecified negligible function, and let ${\sf poly}(\cdot)$ denote an unspecified polynomial.

\subsection{Quantum Computation}
We associate an $n$-qubit quantum system with the complex Hilbert space $\mathcal{H}=\mathbb{C}^{2^n}$. A quantum state $\ket{\psi}\in \mathcal{H}$ is a column vector with norm~1. 
The computational basis of $\mathcal{H}$ is $\{\ket{x}\mid x \in \{0,1\}^{n}\}$, and any quantum state in $\mathcal{H}$ can be written as $\sum_{x \in \{0,1\}^{n}}\alpha_{x}\ket{x}$, where $\alpha_{x} \in \mathbb{C}$ for all~$x$ and $\sum_{x \in \{0,1\}^{n}}|\alpha_{x}|^2=1$. A measurement of a quantum state $\ket{\psi}=\sum_{x}\alpha_{x}\ket{x}$ in the computational basis returns the value $x \in \{0,1\}^{n}$ with probability $|\alpha_{x}|^{2}$; we call the observed $x$ the measurement outcome. If the outcome of measuring $\ket{\psi}$ is $x$, then the measurement collapses 
$\ket{\psi}$ to~$\ket{x}$.
If $\ket{\psi}$ is a column vector, then $\bra{\psi}$ is the row vector obtained by taking the conjugate transpose of~$\ket{\psi}$.

We write $\norm{\ket{\psi}}$ to denote the Euclidean norm of the vector $\ket{\psi}$; i.e., if 
$\ket{\psi}=\sum_{x}\alpha_{x}\ket{x}$ then $\norm{\ket{\psi}}=\sqrt{\sum_x |\alpha_x|^2}$.
The Euclidean distance $\norm{\ket{\phi}-\ket{\psi}}$ gives a measure of how much two quantum states $\ket{\phi},\ket{\psi}$ differ. Another measure of distance between two quantum states is given by the trace distance between their respective density matrices.  For a quantum state $\ket{\psi}$, its corresponding density matrix is $\ket{\psi}\bra{\psi}$. For quantum states $\ket{\phi}$ and $\ket{\psi}$, let $\rho=\ket{\phi}\bra{\phi}$ and $\sigma=\ket{\psi}\bra{\psi}$ be their density matrices. 
The trace distance between $\rho$ and $\sigma$ is given by
$$
\text{TD}(\rho,\sigma)=\frac{1}{2}\text{Tr}\left[\sqrt{(\rho-\sigma)^{\dag}(\rho-\sigma)}\right].
$$ 
The Euclidean distance between two quantum states is an upper bound on the trace distances between their respective density matrices~\cite{10.1007/978-3-030-26951-7_9}.
\begin{lemma}
    \label{lemma:measure} 
    Let $\ket{\phi}$ and $\ket{\psi}$ be quantum states with Euclidean distance $\epsilon$. Then the trace distance between $\ket{\phi}\bra{\phi}$ and $\ket{\psi}\bra{\psi}$ is at most $\epsilon \cdot \sqrt{1-\epsilon^2/4}\leq \epsilon$.
\end{lemma}

Since the trace distance between two density matrices is an upper bound on an algorithm's ability to distinguish the corresponding quantum states, the distinguishing advantage of an algorithm given one of two quantum states is upper bounded by their Euclidean distance.

\medskip\noindent{\bf Quantum algorithms.} 
We consider quantum algorithms having classical input and output. Any such algorithm 
$\mathcal{A}$ consists of a unitary transformation mapping some initial state $\ket{\psi_{{\text{in}}}}$ to a final state $\ket{\psi_{{\text{out}}}}$, followed by a measurement of the final state in the computational basis to produce the output. We write $\mathcal{A}:\{0,1\}^{n}\mapsto \{0,1\}^{m}$ to indicate that an $n$-bit classical input is encoded as the initial quantum state, and the measurement outcome is an $m$-bit string. We write $y\leftarrow \A(x)$ to denote that $y$ is the result of running $\A$ on input $\ket{x}$ and measuring the resulting quantum state.

\medskip\noindent{\bf Oracle algorithms.} 
Algorithm $\A$ with classical oracle access to a function \mbox{$H:\{0,1\}^{n}\mapsto \{0,1\}^{m}$} (written~$\A^H$) can query $x \in \bool^n$ and receive the result $H(x) \in \bool^m$ in a single step.
Giving 
$\A$ quantum oracle access to~$H$ (written~$\mathcal{A}^{\ket{H}}$) means that $\A$ can evaluate the unitary
transformation that maps $\ket{x}\ket{y}\mapsto \ket{x}\ket{y\oplus H(x)}$ in a single step. Concretely, this means that 
$\A$ can query $H$ on a superposition of inputs $\ket{\phi} = \sum_{x \in \bool^{n}}\alpha_{x}\ket{x}$ and receive in return the state $\ket{\psi} = \sum_{x \in \bool^{n}}\alpha_{x}\ket{x}\ket{H(x)}$.\footnote{Technically there are two registers here, 
and $\A$ can query $\ket{\phi} = \sum_{x \in \bool^{n}, y \in \bool^m}\alpha_{x}\ket{x}\ket{y}$ to 
receive in return $\sum_{x \in \bool^{n}, y \in \bool^m}\alpha_{x}\ket{x}\ket{y \oplus H(x)}$, but we often omit the second register for notational clarity.} 
For a quantum query $\ket{\phi}=\sum_{x \in \bool^{n}}\alpha_{x}\ket{x}$,
 we define the {\em query magnitude at~$x$} to be $q_{x}(\ket{\phi}):=|\alpha_{x}|^{2}$.
If $\A^{\ket{H}}$ makes a total of $Q$ queries $\ket{\phi_{1}},\ldots,\ket{\phi_{Q}}$  to 
its oracle~$H$, then 
the {\em total query magnitude at~$x$}~is 
$q^{H}_{x}:=\sum_{i=1}^{Q}q_{x}(\ket{\phi_{i}})$.
The following lemma is one way to formalize the intuition that if the total query magnitude of some algorithm at a point~$x$ is small, then the output of the algorithm cannot depend too much of the value of the oracle at~$x$.

\begin{lemma}[Swapping lemma \cite{nayebi_quantum_2015}]
\label{lemma:swapping}
Let $\ket{\phi_f}$ and $\ket{\phi_g}$ denote the final states of a $Q$-query algorithm $\A$ when given quantum oracle access to $f$ and $g$, respectively. Then
$$
\norm{\ket{\phi_f}-\ket{\phi_g}}\leq \sqrt{Q\cdot \sum_{x\,:\,f(x)\neq g(x)}q^{f}_{x}}.
$$
\end{lemma}

A related result is that if an oracle is reprogrammed via a random process that is not too likely to reprogram any particular point, then a quantum-oracle algorithm interacting with either the original or reprogrammed oracle cannot detect the reprogramming.
The following is a restatement of a result by Alagic et al.~\cite[Lemma~3]{alagic_post-quantum_2022}.

\begin{lemma}[Reprogramming lemma]
\label{lemma:repro1}
Let $\mathcal{D}$ be a distinguisher in the following experiment:
\begin{enumerate}
    \item $\mathcal{D}$ outputs a function $F_{0}=F:\{0,1\}^{n}\mapsto \{0,1\}^{m}$ and a randomized algorithm~$\mathcal{B}$, whose output is a partial function $R:D_{R}\mapsto \{0,1\}^{m}$, where $D_{R}\subseteq \{0,1\}^{n}$. Let 
    $$
    \epsilon = \max_{x \in \{0,1\}^{n}}\left\{\Pr_{R\leftarrow \mathcal{B}}[x \in D_{R}]\right\}.
    $$
    \item $\mathcal{B}$ is run to obtain $R$. Let $F_{1}=F_{0}^{(R)}$. A uniform bit $b$ is chosen, and $\mathcal{D}$ is given quantum access to $F_{b}$. 
    \item $\mathcal{D}$ loses access to $F_{b}$, and receives the randomness $r$ used to invoke $\mathcal{B}$ in phase 2. Then $\mathcal{D}$ outputs a guess $b'$.
\end{enumerate}
For any $\mathcal{D}$ making $Q$ queries, it holds that 
$$
\left|\Pr[\mathcal{D} \text{ outputs } 1\mid b=1]-\Pr[\mathcal{D} \text{ outputs }1\mid b=0]\right|\leq 2\cdot Q\cdot \sqrt{\epsilon}.
$$
\end{lemma}

\medskip\noindent{\bf Pseudo-deterministic algorithms.}
Quantum algorithms whose outputs are (close to) deterministic will be of particular interest to us. 
We reformulate the definition to incorporate  a parameter~$\delta$ denoting the probability with which the algorithm returns a ``wrong'' value.

\begin{definition} 
Let $\A:\{0,1\}^{n}\mapsto \{0,1\}^{m}$ be a quantum algorithm.
We say $\A$ is {\bf $\delta$-deterministic} (for $\delta>1/2$) if for all $x \in \{0,1\}^{n}$ there exists $y_x \in \{0,1\}^{m}$ such that 
$$
    \Pr[y \leftarrow \A(x) : y\neq y_x]\leq \delta.
$$
\end{definition}


\noindent
When two $\delta$-deterministic algorithms $\A$ and $\B$ agree on their most likely output on some input $x$, we write $\A(x)\qeq{\delta}\B(x)$. A quantum-oracle algorithm $\mathcal{A}$ is $\delta$-deterministic if $\mathcal{A}^{\ket{f}}$ is $\delta$-deterministic for every oracle~$f$. (We stress that in this case the value $y_x$ may depend on~$f$.)

\subsection{Pseudorandom Generators}

Roughly, a pseudorandom generator $G$ takes as input a uniform seed $s \in \bool^k$ and outputs a longer string $g \in \bool^\ell$ that is indistinguishable from a uniform $\ell$-bit string. 
As we are interested in the case where $G$ has  oracle access to a random oracle, we specialize our definition to that setting.
Formally, the distinguishing advantage of $\A$ relative to~$G$ is
$$
    {\sf Adv}_{\A,G}^{\sf PRG} = \left|\Pr_{\substack{s \leftarrow \bool^k\\ H \leftarrow {\sf Func}_{n,m}}}\left[\A(G(s))=1\right]-\Pr_{\substack{g \leftarrow \bool^\ell\\ H \leftarrow {\sf Func}_{n,m}}}\left[\A(g)=1\right]\right|,
    $$
where both $\A$ and $G$ are given oracle access to~$H$. If $\A$ is given classical (resp., quantum) access to~$H$, we may emphasize this by referring to the \emph{classical (resp., quantum) distinguishing advantage of~$\A$ relative to~$G$}.
Note that $G$ might have either classical or quantum access to~$H$.

\section{Lifting Theorem for Classical PRGs}

Fix an algorithm $G:\{0,1\}^{k}\mapsto \{0,1\}^{\ell}$ that  always makes exactly $Q_{G}$  classical queries to an oracle~$H:\bool^n\mapsto\bool^m$.
We show (cf.\ Theorem~\ref{theorem:lifting1}) that for any algorithm $\A$ making $Q_\A$ \emph{quantum} queries to~$H$ there is an algorithm~$\B$ making \emph{classical} queries to~$H$ such that (1)~the  distinguishing advantage of $\B$ relative to~$G$ is close to the  distinguishing advantage of $\A$ relative to~$G$, and (2)~the number of queries $\B$ makes to its oracle is polynomial in $k$, $n$, $Q_G$, $Q_\A$, and the inverse of the distinguishing advantage of~$\A$.

We begin by introducing some notation and definitions.
For notational convenience, we overload the definition of $G$ so that in addition to its actual output $g \in \bool^\ell$ it also outputs the list $\tau$ of queries it made to $H$ and the associated answers.
We let $\mathcal{G}$ denote the distribution on $(H, s, g, \tau)$ generated by sampling a uniform function~$H\in {\sf Func}_{n,m}$, choosing a uniform seed $s \in \bool^k$, and then running $G^H(s)$ to generate~$g$ and~$\tau$.
We also define the following conditional distributions:
\begin{itemize}
    \item For $g\in \bool^\ell$ in the range of the PRG, let $\mathcal{T}_{g}$ be the distribution on $\tau$ after the above process conditioned on $g$ being the output of the PRG. Similarly, for a partial function~$h$ such that there exists $(H,s)$ with $H \in {\sf Func}_{n,m}(h)$ satisfying $g = G^{H}(s)$, let $\mathcal{T}_{g,h}$ be the distribution on $\tau$ conditioned on (i)~$g$ being the output of the PRG and (ii)~$H \in {\sf Func}_{n,m}(h)$. 

    \item  We write ${\mathcal{O}}$ for the uniform distribution over ${\sf Func}_{n,m}$. For $g$ in the range of the PRG, we let $\mathcal{O}_{g}$ denote the distribution over ${\sf Func}_{n,m}$ conditioned on $g$ being the output of the PRG. Similarly, for a partial function $h$ such that there exists $(H,s)$ with $H \in {\sf Func}_{n,m}(h)$ satisfying $g = G^{H}(s)$,  let $\mathcal{O}_{g,h}$ be the distribution over ${\sf Func}_{n,m}$ conditioned on (i)~$g$ being the output of the PRG and (ii)~$H \in {\sf Func}_{n,m}(h)$. 
\end{itemize}

\def\O{\mathcal{O}}
\def\T{\mathcal{T}}
\def\func{{\sf Func}}
\def\comp{{\sf Func}_{n,m}}

Writing everything out explicitly, we have
$$
{\sf Adv}_{\A,G}^{\sf PRG} = \big|\Pr[{\sf PRG}_{\A,G}=1]-\Pr[{\sf Rand}_{\A,G}=1]\big|
$$
where 
${\sf PRG}_{\A,G}$ and ${\sf Rand}_{\A,G}$ are defined as follows:

\noindent\begin{minipage}{\textwidth}
  \centering
\begin{minipage}{.45\textwidth}
    \begin{algorithm}[H]
    \renewcommand{\thealgorithm}{}
    \floatname{algorithm}{}
    \begin{algorithmic}[1]
    \State $H\leftarrow \mathcal{O}$
    \State $s\leftarrow \{0,1\}^{k}$
    \State $(g,\tau):=G^{H}(s)$
    \State  $b\leftarrow\A^{\ket{H}}(g)$ 
    \State return $b$
    \end{algorithmic}
      \caption{${\sf PRG}_{\A,G}$} 
    \label{alg:QPRG}
\end{algorithm} 
  \end{minipage}
\begin{minipage}{.45\textwidth}
    \begin{algorithm}[H]
    \renewcommand{\thealgorithm}{}
    \floatname{algorithm}{}
    \begin{algorithmic}[1]
    \State $H\leftarrow \mathcal{O}$
    \State $g\leftarrow \{0,1\}^{\ell}$
    \State $b\leftarrow \A^{\ket{H}}(g)$ 
    \State return $b$
    \end{algorithmic}
      \caption{${\sf Rand}_{\A,G}$} 
    \label{alg:QRand}
\end{algorithm} 
  \end{minipage}
\end{minipage}
\vskip.8\baselineskip
\noindent

\noindent
We also define versions of these experiments where we condition on a particular value~$g$:

\noindent\begin{minipage}{\textwidth}
  \centering
\begin{minipage}{.45\textwidth}
     \begin{algorithm}[H]
     \renewcommand{\thealgorithm}{}
     \floatname{algorithm}{}
    \begin{algorithmic}[1]
    \State $H\leftarrow \mathcal{O}_{g}$
    \State $b\leftarrow\A^{\ket{H}}(g)$ 
    \State return $b$
    \end{algorithmic}
      \caption{${\sf PRG}_{\A,G}(g)$} 
    \label{alg:QPRG_g}
\end{algorithm} 
  \end{minipage}
  \begin{minipage}{.45\textwidth}
    \begin{algorithm}[H]
    \renewcommand{\thealgorithm}{}
    \floatname{algorithm}{}
    \begin{algorithmic}[1]
    \State $H\leftarrow \mathcal{O}$
    \State $b\leftarrow \A^{\ket{H}}(g)$ 
    \State return $b$
    \end{algorithmic}
      \caption{${\sf Rand}_{\A,G}(g)$}
    \label{alg:QRand_g}
\end{algorithm} 
  \end{minipage}
\end{minipage}
\vskip.8\baselineskip
\noindent
It is immediate that 
\[\Pr[{\sf PRG}_{\A,G}=1] = \Pr[(H, s, g, \tau) \leftarrow \mathcal{G}: {\sf PRG}_{\A,G}(g)=1].\]

In Figure~\ref{fig:ClasssicalDist} we describe an algorithm $\B$ that makes only \emph{classical} queries to~$H$. 
On input $g$, algorithm~$\B$ first runs a subroutine ${\sf findTranscript}^{H}$ that, intuitively, results in $\B$ querying $H$ on all points ``likely'' to appear in the transcript of $G$'s interaction with~$H$ (conditioned on $g$ being the output of~$G$).
The output $h$ of ${\sf findTranscript}^{H}$ contains
all query/answer pairs that $\B$ thus learns.
$\B$ then chooses a random function $H'$ consistent with the partial function~$h$, and runs $\A$ with input~$g$ and access to~$H'$. (Note that this can be done with no further queries to~$H$.)
Finally, $\B$ outputs whatever $\A$ does.

Our aim is to show that the distribution of the output of $\B^H(g)$ is close to the distribution of the output of $\A^H(g)$ regardless of whether $g$ is a uniform string or whether $g$ was the output of $G^H(s)$ on a uniform seed~$s$. This is easy to show when $g$ is uniform, as in that case $g$ is uncorrelated with the random oracle~$H$ and so the distinguishing advantage of $\A$ is unchanged regardless of whether it interacts with oracle $H$ or an independent random oracle~$H'$.

The more challenging case is when $g$ is the output of the PRG. Intuitively we can still resample the values of the oracle which were not queried by the PRG to produce $g$, as the computation of the PRG is independent of these values. This leaves us with the problem of simulating the values of the oracle which are now correlated with $g$. For each point on which the oracle $H$ was queried by $G$, we can categorize the point based on if it had a high (non negligible) probability of being queried by $G$, where the probability is over the randomness of the PRG input $s$ as well as the randomness over the previous oracle values which were queried. To handle points which are likely to be queried, our distinguisher simply queries all points which have a sufficiently large probability of showing up in a transcript between the oracle and the PRG. This process is formalized in Algorithm~\ref{alg:findTranscript}, and we prove its efficiency in Lemma~\ref{lemma:findTranscriptProof}. To handle the remaining points in the transcript which are not likely to have been queried, we rely on an oracle reprogramming result due to Alagic et al. \cite{alagic_post-quantum_2022} (Lemma~\ref{lemma:repro1}), which states that if an oracle is reprogrammed by some randomized process then a quantum distinguisher interacting with either the original oracle or the reprogrammed oracle will not notice the reprogramming provided that no particular point is very likely to be reprogrammed. Viewing the randomized reprogramming process as sampling a transcript between the oracle and PRG, and then only outputting the points which were not likely to have been queried, we can argue that the distinguisher will not notice if we resample these points randomly. 

\noindent\begin{minipage}{\textwidth}
  \centering
  \begin{minipage}{.41\textwidth}
     \begin{algorithm}[H]
    \begin{algorithmic}[1]
    \If{$\not \exists s,H$ such that $G^{H}(s)=g$}
    \State return $0$ 
    \EndIf 
    \State $\delta = \left(\frac{{\sf Adv}_{\A,G}^{\sf PRG}}{6\cdot Q_{\A}}\right)^{2}$
    \State $h={\sf findTranscript}^{H}(g,\delta)$
    \IIf{$h=\perp$} return 0 \EndIIf
    \State $H'\leftarrow {\sf Func}_{n,m}(h)$
    \State $b\leftarrow \A^{\ket{H'}}(g)$
    \State return $b$
    \end{algorithmic}
      \caption{${\B}^{H}(g)$}
    \label{alg:B}
\end{algorithm} 
  \end{minipage}
  \begin{minipage}{.51\textwidth}
     \begin{algorithm}[H]
    \begin{algorithmic}[1]
    \State ${\sf limit}=\ceil*{\frac{-\ln(\delta)\cdot 4Q_{G}^{2}}{\delta^{2}}}+1$
    \State $i = 0$, $h = \emptyset$, $\epsilon = 1$
    \While{$\epsilon>\delta $ and $i<{\sf limit}$ }
    \State $i = i + 1$
    \IIf{$\mathcal{T}_{g,h}$ is undefined} return $\perp$\EndIIf
    \State $x' = \arg\max_{x}\{\Pr_{\tau\leftarrow \mathcal{T}_{g,h}}[(x,\star) \in \tau]\}$
    \State Query $H$ on $x'$ and add $(x',H(x'))$ to $h$.
    \State $\epsilon = \max_{x}\{\Pr_{\tau\leftarrow \mathcal{T}_{g,h}}[(x,\star) \in \tau]\}$
    \EndWhile 
    \IIf{$i\geq {\sf limit}$} return $\perp$ \EndIIf
    \State return $h$
    \end{algorithmic}
      \caption{${\sf findTranscript}^{H}(g,\delta)$}
    \label{alg:findTranscript}
\end{algorithm} 
  \end{minipage}
  \captionof{figure}{}
  \label{fig:ClasssicalDist}
\end{minipage}
\vskip.8\baselineskip
\noindent


We now give a more technical sketch of the proof. As just discussed, the bulk of the technical work is in showing that $\B^{(\cdot)}(g)$ (Algorithm~\ref{alg:B}) correctly simulates the behavior of the quantum algorithm in the experiment ${\sf PRG}_{\A,G}$ (Lemma~\ref{lemma:mainLemma}). Since the input $g$ given to the distinguisher is fixed in ${\sf PRG}_{\A,G}(g)$, the distribution over outputs from the distinguisher only depends on the distribution from which the oracle is drawn. With this in mind, our goal will be to find a distribution $\mathcal{O}_{g}'$ over oracles such that~(1) the output of $\A$ is distributed almost identically when given an oracle from $\mathcal{O}'_{g}$ or from $\mathcal{O}_{g}$, and~(2) there is some way of sampling an oracle from $\mathcal{O}'_{g}$ without making a large number of queries to the original oracle. 

We now introduce three hybrid experiments, each of which with a slightly different method of generating the oracle that $\A$ interacts with. 
\begin{itemize}
    \item ${\sf Hyb}_1$: In this experiment we sample an oracle $H\leftarrow \mathcal{O}_{g}$, run ${\sf findTranscript}^{H}(g,\delta)$ on that oracle to obtain a partial function $h$, and then resample the remaining values of the oracle that have not been queried by ${\sf findTranscript}$ from the distribution $\mathcal{O}_{g,h}$. 
    \item ${\sf Hyb}_2$: This experiment replaces the sampling procedure $H\leftarrow \mathcal{O}_{g,h}$ with a procedure in which we sample a transcript $\tau\leftarrow \mathcal{T}_{g,h}$, extend the domain of the partial function $h$ by uniformly generating the remaining values to produce $H_0\leftarrow {\sf Func}_{n,m}(h)$. Finally we reprogram $H_0$ according to the partial function $\tau\setminus h$ to produce $H_1=H_{0}^{(\tau\setminus h)}$. In Lemma~\ref{lemma:transition1} we prove that this procedure is equivalent to the one in ${\sf Hyb}_1$. 
    \item ${\sf Hyb}_3$: In this experiment, we omit the final reprogramming step of ${\sf Hyb}_2$ and give $\A$ access to $H_0$ as defined above. To compare ${\sf Hyb}_2$ and ${\sf Hyb}_3$, note that the task of distinguishing $H_0$ from $H_1$ matches the reprogramming game of Lemma~\ref{lemma:repro1}, where the randomized reprogramming process from the lemma outputs the partial function~$\tau\setminus h$. Lemma~\ref{lemma:repro1} gives a bound on distinguishing $H_0$ and $H_1$ based on the largest probability of any particular point showing up in the domain of $\tau\setminus h$. In Lemma~\ref{lemma:findTranscriptProof} we obtain such a bound by showing that with high probability over the partial function $h$ returned by ${\sf findTranscript}$, no point has a high probability of showing up in the domain of~$\tau\setminus h$. 
\end{itemize}

We now turn to the proof. 
We rely on the following result of Austrin et al.~\cite[Lemma~3.6]{quantumKeyAgreement}.

\begin{lemma}
\label{lemma:heavyQuery}
    Let $L, z_1,x_1,..., z_q, x_q$ be a finite sequence of correlated random variables, where $L$ ranges over subsets of a universe $\mathcal{U}$ and 
    $x_i \in \mathcal{U}$ for all~$i$, and with the property that if $i\neq j$ then $x_i\neq x_j$. Conditioned on a sequence $z_1,x_1,...,z_q$, say that $x_i$ is $\delta$-heavy if $\Pr[x_i \in L\mid z_1,x_1,...,z_i]\geq \delta$ and, for the same sequence, define $\mathcal{S}=\{x_i\mid x_i \text{ is }\delta\text{-heavy}\}$. Then, $\mathbb{E}[|S|]\leq \mathbb{E}[|L|]/\delta$.
\end{lemma}

We start by proving an important property about the ${\sf findTranscript}$ subroutine.

\begin{lemma}
\label{lemma:findTranscriptProof}
    Let $\delta \in (0,1)$ and $g \in \bool^\ell$. Then with probability at least $1-\delta$ (over the randomness of the oracle $H$ drawn from $\mathcal{O}_{g}$), ${\sf findTranscript}^{H}(g,\delta)$ returns a partial function $h$ with domain of size at most $\frac{Q_{G}}{\delta^{2}}$ such that 
    $$
        \max_{x\not \in D_{h}}\left\{\Pr_{\tau \leftarrow \mathcal{T}_{g,h}}[(x,\star) \in \tau]\right\}\leq \delta.
    $$    
\end{lemma}


\begin{proof}
    The final inequality in the lemma holds by  the termination condition for ${\sf findTranscript}$, so it only remains to prove that the size of the partial function $h$ is at most $Q_{G}^{2}/\delta$. Let $L$  be the random variable describing a transcript $\tau\leftarrow \mathcal{T}_{g}$ of the queries $G$ makes to the random oracle, and let $z_1=g$ be the output of the PRG. For $i\geq 1$, let the random variable $x_i$ be the $i$th query made by ${\sf findTranscript}$ to the random oracle, and let $z_{i+1}$ be the value returned by the oracle. Letting $q$ be the number of queries made by ${\sf findTranscript}$, set $x_{i}=\perp$ and $z_{i+1}=0$ for $q<i\leq 2^{n}$. Every query made by ${\sf findTranscript}$ is $\delta$-heavy, and so  Lemma~\ref{lemma:heavyQuery} implies $\mathbb{E}[|\mathcal{S}|]\leq Q_{G}/\delta$. It follows from Markov's inequality that $\Pr\left[|\mathcal{S}|>Q_{g}/\delta^{2}\right]\leq \delta$. Since the partial function $h=\{(x_{i},z_{i+1})\}_{i=1}^{q}$ returned by ${\sf findTranscript}$ has $\mathcal{S}$ as its domain, the lemma statement follows.
\end{proof}

Now we turn to the proof of Lemma~\ref{lemma:transition1} which justifies the transition from ${\sf Hyb}_{1}(g)$ to~${\sf Hyb}_{2}(g)$.

\begin{lemma} 
\label{lemma:transition1}
Let $g_0 \in \{0,1\}^{\ell}$, and let $h_0$ be a partial function such that there exists $(s,H)$ where $H \in {\sf Func}_{n,m}(h_0)$ and $(g_0,\tau) = G^{H}(s)$ for some $\tau$. Define the following random variables: 
    \[{\sf O}'_{g_0,h_0}:=\left\{\begin{array}{c}
    \tau\leftarrow \mathcal{T}_{g_0,h_0}\\
    H\leftarrow {\sf Func}_{n,m}(h_0)
 \end{array}: H^{(\tau\setminus h_0)}\right\}  \quad \text{ and }\quad {\sf O}_{g_0,h_0}:=\left\{\begin{array}{c}
    H\leftarrow \mathcal{O}_{g_0,h_0}
 \end{array}: H\right\}. \]
 Then ${\sf O}'_{g_0,h_0}$ and ${\sf O}_{g_0,h_0}$ are distributed identically. 
\end{lemma}

\begin{proof}
Let ${\mathcal{U}}$ be the uniform distribution over oracle-seed combinations in ${\sf Func}_{n,m}\times \bool^{k}$, and let $R_{g,\tau,h}$ be a relation over oracle-seed pairs such that 
$$
(H,s) \in R_{g,\tau,h} \iff H \in {\sf Func}(h) \land (g,\tau) = G^{H}(s).
$$
First, we can write the probability of ${\sf O}'_{g_0,h_0}$ generating a particular function $H$ as
\begin{align}
    \Pr[{\sf O}'_{g_0,h_0}=H]&=\sum_{\tau}\Pr[\tau \leftarrow \mathcal{T}_{g_0,h_0}]\cdot \Pr_{{\sf O}\leftarrow {\sf Func}_{n,m}(h_0\cup \tau)}[{\sf O}=H]. \label{eq:OracleExpression1}
\end{align}
On the other hand, we can write the probability of the random variable ${\sf O}_{g_0,h_0}$ generating some function $H$ as
\begin{align}
    \Pr[{\sf O}_{g_0,h_0}=H]&=\Pr_{{\sf O},{\sf S}\leftarrow \mathcal{U}}[{\sf O}\in {\sf Func}(H\setminus h_0)\mid {\sf O}\in {\sf Func}(h_0), (g_0,\star)=G^{\sf O}({\sf S})]\nonumber\\
    &=\sum_{\tau}\Pr[\tau\leftarrow \mathcal{T}_{g_0,h_0}]\cdot \Pr_{{\sf O},{\sf S}\leftarrow \mathcal{U}}[{\sf O}\in {\sf Func}(H\setminus h_0)\mid {\sf O}\in {\sf Func}(h),(g_0,\tau) = G^{\sf O}({\sf S})]\nonumber\\
    &=\sum_{\tau}\Pr[\tau\leftarrow \mathcal{T}_{g_0,h_0}]\cdot \Pr_{{\sf O},{\sf S}\leftarrow \mathcal{U}}[{\sf O}\in {\sf Func}(H\setminus h_0)\mid ({\sf O},{\sf S}) \in R_{g_0,\tau,h_0}]\label{eq:OracleExpression2}
\end{align}
Comparing Equation~(\ref{eq:OracleExpression1}) above with Equation~(\ref{eq:OracleExpression2}), we will prove our result by showing that for any $\tau$ that is consistent with $h_0$,
$$
\Pr_{{\sf O},{\sf S}\leftarrow \mathcal{U}}[{\sf O}\in {\sf Func}\left(H\setminus h_0\right)\mid ({\sf O},{\sf S}) \in R_{g_0,\tau,h_0}] = \Pr_{{\sf O}\leftarrow {\sf Func}_{n,m}(h_0\cup \tau)}[{\sf O}=H].
$$
To prove the above, we will show that for any two oracles $H_0,H_1 \in {\sf Func}(h_0\cup \tau)$,
$$
\Pr_{{\sf O},{\sf S}\leftarrow \mathcal{U}}[{\sf O} = H_0\mid ({\sf O},{\sf S})\in R_{g_0,\tau,h_0}]=\Pr_{{\sf O},{\sf S}\leftarrow \mathcal{U}}[{\sf O} = H_1\mid ({\sf O},{\sf S})\in R_{g_0,\tau,h_0}].
$$
Let $\tau$ be any partial function consistent with $h_0$, and consider two oracle seed pairs $(H_0,s_0)$ and $(H_1,s_1)$ such that both pairs are in the relation $R_{g_0,\tau,h_0}$. Then we have 
$$
\Pr_{{\sf O},{\sf S}\leftarrow \mathcal{U}}[({\sf O},{\sf S})=(H_{0},s_{0})\mid ({\sf O},{\sf S}) \in R_{g_0,\tau,h_0}] = \Pr_{{\sf O},{\sf S}\leftarrow \mathcal{U}}[({\sf O},{\sf S})=(H_{1},s_{1})\mid ({\sf O},{\sf S}) \in R_{g_0,\tau,h_0}].
$$
Define the following set:
$$
{\sf Good}_{g_0,\tau}(H)=\{s\mid G^{H}(s)=(g_0,\tau)\}.
$$
Since $H_{0},H_{1} \in {\sf Func}(\tau)$, we have that ${\sf Good}_{g_0,\tau}(H_{0})={\sf Good}_{g_0,\tau}(H_{1})$. To see this, note that since the execution of $G$ on input $s \in {\sf Good}_{g_0,\tau}(H_0)$ with oracle $H_0$ is independent of oracle values outside of $\tau$, it follows that the execution on $G$ with oracle $H_1$ will be identical. It follows that 
\begin{align*}
    \Pr_{{\sf O},{\sf S}\leftarrow \mathcal{U}}[{\sf O}=H_{0}\mid ({\sf O},{\sf S}) \in R_{g_0,\tau,h_0}]&=\sum_{s \in {\sf Good}_{g_0,\tau}(H_{0})}\Pr_{{\sf O},{\sf S}\leftarrow \mathcal{U}}[({\sf O},{\sf S})=(H_{0},s)\mid ({\sf O},{\sf S}) \in R_{g_0,\tau,h_0}]\\
    & = \sum_{s \in {\sf Good}_{g_0,\tau}(H_{1})}\Pr_{{\sf O},{\sf S}\leftarrow \mathcal{U}}[({\sf O},{\sf S})=(H_{1},s)\mid ({\sf O},{\sf S}) \in R_{g_0,\tau,h_0}]\\
    & = \Pr_{{\sf O},{\sf S}\leftarrow \mathcal{U}}[{\sf O}=H_{1}\mid ({\sf O},{\sf S}) \in R_{g_0,\tau,h_0}].
\end{align*}
This concludes the proof.
\end{proof}

With the above two results, we are now ready to proceed to the main proof. We start by showing that our classical distinguisher $\B^{H}$ matches the performance of the quantum distinguisher $\A^{\ket{H}}$ when interacting with a PRG.

\begin{lemma}
\label{lemma:mainLemma}
Let $\A$ be a quantum algorithm making $Q_{\A}$ quantum oracle queries, and let $\B^{H}$ be as described in Algorithm~\ref{alg:B}. Then, 
    $$
        \big|\Pr[{\sf PRG}_{\A,G}=1]-\Pr[{\sf PRG}_{\B,G}=1]\big|\leq  \frac{1}{2}{\sf Adv}_{\A,G}^{\sf PRG}.
    $$
\end{lemma}
\begin{proof}
Let 
$\epsilon = \max_{g}\left\{\big|\Pr[{\sf PRG}_{\A,G}(g)=1]-\Pr[{\sf PRG}_{\B,G}(g)=1]\big|\right\}.$ Then we have 
    \begin{align*}
        \big|\Pr[{\sf PRG}_{\A,G}=1]-&\Pr[{\sf PRG}_{\B,G}=1]\big| \\
        &=\left|\sum_{g}\Pr_{\substack{s \leftarrow \bool^k\\ H \leftarrow {\sf Func}_{n,m}}}[g=G^{H}(s)]\big(\Pr[{\sf PRG}_{\A,G}(g)=1]-\Pr[{\sf PRG}_{\B,G}(g)=1]\big)\right|\\
        &\leq \epsilon \cdot \sum_{g}\Pr_{\substack{s \leftarrow \bool^k\\ H \leftarrow {\sf Func}_{n,m}}}[g=G^{H}(s)] \; = \; \epsilon.
    \end{align*}
    We now derive a bound for $\epsilon$ via the sequence of games ${\sf Hyb}_1$, ${\sf Hyb}_2$, ${\sf Hyb}_3$, discussed earlier. Let $g$ be in the range of the PRG (for at least one seed-oracle combination $s,H$). The pseudocode for the games is presented here. Note that the games only differ in how the oracle is generated. Thus, our task will be to show that in each transition between games, the quantum distinguisher $\A$ will not notice the difference between oracles. 

    \noindent\begin{minipage}{\textwidth}
  \centering
    \begin{minipage}{.33\textwidth}
    \begin{algorithm}[H]
    \renewcommand{\thealgorithm}{}
    \floatname{algorithm}{}
    \begin{algorithmic}[1]
    \State $H\leftarrow \mathcal{O}_{g}$
    \State $h={\sf findTranscript}^{H}(g,\delta)$
    \State $H_{1}\leftarrow \mathcal{O}_{g,h}$
    \State $b\leftarrow \A^{\ket{H_{1}}}(g)$
    \State return $b$
    \end{algorithmic}
      \caption{${\sf Hyb}_{1}(g)$}
    \label{alg:reprogram}
\end{algorithm} 
  \end{minipage} 
  \begin{minipage}{.4\textwidth}
     \begin{algorithm}[H]
     \renewcommand{\thealgorithm}{}
    \floatname{algorithm}{}
    \begin{algorithmic}[1]
    \State $H\leftarrow \mathcal{O}_{g}$
    \State $h={\sf findTranscript}^{H}(g,\delta)$
    \State $H_{0}\leftarrow {\sf Func}_{n,m}(h)$
    \State $\tau\leftarrow \mathcal{T}_{g,h}$
    \State $H_{1}=H_{0}^{(\tau\setminus h)}$
    \State $b\leftarrow \A^{\ket{H_{1}}}(g)$ \fbox{$b\leftarrow \A^{\ket{H_{0}}}(g)$}
    \State return $b$
    \end{algorithmic}
      \caption{${\sf Hyb}_{2}(g)$ \fbox{${\sf Hyb}_{3}(g)$}}
    \label{alg:reprogram}
\end{algorithm} 
  \end{minipage}
\end{minipage}
\vskip.8\baselineskip
\noindent 
    \medskip\noindent{\bf (${\sf PRG}_{\A,G}(g)\mapsto {\sf Hyb}_{1}(g)$):} We show that $H$ in ${\sf PRG}_{\A,G}(g)$ is distributed identically to $H_{1}$ in ${\sf Hyb}_{1}$. The only property of ${\sf findTranscript}$ that we require is that the algorithm outputs every query that it makes. Fix some oracle $H$, and let $h={\sf findTranscript}^{H}(g,\delta)$. We will show that $H$ has the same probability of being returned in each of the games. First, we have 
    \begin{align*}
         \Pr_{\sf O\leftarrow \mathcal{O}_{g}}[{\sf O}=H]&=\Pr_{{\sf O}\leftarrow \mathcal{O}_{g}}[{\sf O}\in {\sf Func}_{n,m}(h)\land {\sf O}\in {\sf Func}_{n,m}(H\setminus h)]\\
         &=\Pr_{{\sf O}\leftarrow \mathcal{O}_{g}}[{\sf O}\in {\sf Func}_{n,m}(h)]\cdot \Pr_{{\sf O}\leftarrow \mathcal{O}_{g}}[{\sf O}\in {\sf Func}_{n,m}(H\setminus h)\mid {\sf O}\in {\sf Func}_{n,m}(h)]\\
         &=\Pr_{{\sf O}\leftarrow \mathcal{O}_{g}}[{\sf O}\in {\sf Func}_{n,m}(h)]\cdot \Pr_{{\sf O}\leftarrow \mathcal{O}_{g,h}}[{\sf O}\in {\sf Func}_{n,m}(H\setminus h)]
    \end{align*}
    The left hand side of the above equation is the probability of ${\sf PRG}_{\A,G}(g)$ generating $H$ on line 1. We will now argue that the right hand side gives the probability of ${\sf Hyb}_{1}(g)$ generating $H$ on line 3. Since the partial function $h$ output by ${\sf findTranscript}$ contains all oracle values that were queried, $h$ is independent of any oracle values outside of the domain of $h$. Therefore, for any $H'\neq H$, 
    $$
        H'\in {\sf Func}_{n,m}(h)\iff h={\sf findTranscript}^{H'}(g,\delta).
    $$
    It follows that $\Pr_{{\sf O}\leftarrow \mathcal{O}_{g}}[O \in {\sf Func}_{n,m}(h)]=\Pr_{{\sf O}\leftarrow \mathcal{O}_{g}}[h={\sf findTranscript}^{{\sf{O}}}(g,\delta)]$. Therefore the right hand side of the previous calculation gives the probability of $H$ being generated on line 3 of ${\sf Hyb}_{1}(g)$.
    
    \medskip\noindent{\bf (${\sf Hyb}_{1}(g)\mapsto {\sf Hyb}_{2}(g)$):} It follows from Lemma~\ref{lemma:transition1} that $H_{1}$ on line 5 of ${\sf Hyb}_{2}(g)$ is distributed identically to $H_{1}$ on line 3 of ${\sf Hyb}_{1}(g)$. It follows that $\Pr[{\sf Hyb}_{1}(g)=1]=\Pr[{\sf Hyb}_{2}(g)=1]$.
    
    \medskip\noindent{\bf (${\sf Hyb}_{2}(g)\mapsto {\sf Hyb}_{3}(g)$):} We will map the task of distinguishing ${\sf Hyb}_2$ from ${\sf Hyb}_3$ onto the distinguishing game of Lemma~\ref{lemma:repro1}. Referring to the notation in Lemma~\ref{lemma:repro1}, let $F_{0}$ represent $H_0$ on line 3 of ${\sf Hyb}_2$ (and~${\sf Hyb}_3$). Define the randomized reprogramming algorithm $\mathcal{B}$ so that it samples $\tau\leftarrow \mathcal{T}_{g,h}$ and produces $F_1=F_{0}^{(\tau\setminus h)}$ (lines 4--5 of ${\sf Hyb}_2$ and ${\sf Hyb}_3$). Since distinguishing $F_0$ and $F_1$ (or equivalently $H_0$ and $H_1$) is equivalent to distinguishing the two games, it follows that 
    $$
        \big|\Pr[{\sf Hyb}_{2}(g)=1]-\Pr[{\sf Hyb}_{3}(g)=1]\big|\leq 2\cdot Q_{A}\cdot \sqrt{\max_{x \not \in D_{h}}\left\{\Pr_{\tau\leftarrow \mathcal{T}_{g,h}}[(x,\star) \in \tau]\right\}},
    $$
    where we recall that $D_{h}$ denotes the domain of the partial function $h$. Lemma~\ref{lemma:findTranscriptProof} tells us that with probability at least $1-\delta$ over the function $h$, the distribution $\mathcal{T}_{g,h}$ on line 4 of ${\sf Hyb}_{2}(g)$ has the property that 
    $$
         \max_{x\not \in D_{h}}\left\{\Pr_{\tau\leftarrow \mathcal{T}_{g,h}}[(x,\star) \in \tau]\right\}\leq \delta.
    $$
   It follows that ${\sf Hyb}_{2}(g)$ and ${\sf Hyb}_{3}(g)$ can be distinguished with advantage at most $(1-\delta)\cdot 2 Q_{A}\sqrt{\delta}+\delta$, where $Q_{\A}$ is the number of oracle queries made by $\A$. Therefore we have 
    $$
    \big|\Pr[{\sf Hyb}_{2}(g)=1]-\Pr[{\sf Hyb}_{3}(g)=1]\big|\leq 2\cdot Q_{\A}\cdot \sqrt{\delta}+\delta.
    $$

    \medskip\noindent{\bf (${\sf Hyb}_{3}(g)\mapsto {\sf PRG}_{\B,G}(g)$):} These two games are identical as $\B$ can generate $H_0$ on line 3 itself and is therefore able to simulate the entire computation of $A^{\ket{H_0}}(g)$ including the superposition queries. Therefore,  $\Pr[{\sf Hyb_{3}}(g)=1]=\Pr[{\sf PRG}_{\B,G}(g)=1]$. Putting everything together, we have 
    $$
        \big|\Pr[{\sf PRG}_{\A,G}(g)=1]-\Pr[{\sf PRG}_{\B,G}(g)=1]\big|\leq 2Q_{\A}\sqrt{\delta}+\delta. 
    $$
    Letting $\delta = \left(\frac{{\sf Adv}_{\A,G}^{\sf PRG}}{6\cdot Q_{\A}}\right)^{2}$ as in Algorithm~\ref{alg:B}, we obtain the bound in the lemma statement. \end{proof}

We can now prove the main result of this section.

\begin{theorem}
\label{theorem:lifting1}
    Let $G:\bool^{k}\mapsto \bool^{\ell}$ be a deterministic algorithm making $Q_{G}$ classical queries to an oracle $H:\bool^{n}\mapsto \bool^{m}$, and let $\A$ be an algorithm making $Q_{\A}$ quantum queries to~$H$. Let $\epsilon={\sf Adv}^{{\sf PRG}}_{\A,G}$.
 Then there is an algorithm $\B$ making ${\sf poly}(k,n,Q_{G},Q_{\A},\epsilon^{-1})$ classical queries to~$H$ with ${\sf Adv}^{{\sf PRG}}_{\B,G} \geq \epsilon/2$.
\end{theorem}
\begin{proof}
Let $\A$ be a quantum distinguisher and let $\B$ be Algorithm~2. Let
$$
\delta_{{\sf PRG}} = \Pr[{\sf PRG}_{\A,G}=1]-\Pr[{\sf PRG}_{\B,G}=1], \quad \quad \text{and}\quad \quad \delta_{{\sf rand}}=\Pr[{\sf Rand}_{\B,G}=1]-\Pr[{\sf Rand}_{\A,G}=1].
$$
Then 
    \begin{align*}
        {\sf Adv}_{\B,G}^{\sf PRG}& = \left|\Pr[{\sf PRG}_{\B,G}=1]-\Pr[{\sf Rand}_{\B,G}=1]\right|\\
        & = \left|\Pr[{\sf PRG}_{\A,G}=1]-\delta_{{\sf PRG}}-\Pr[{\sf Rand}_{\A,G}=1] -\delta_{{\sf rand}}\right|\\
        &=\left|\Pr[{\sf PRG}_{\A,G}=1]-\Pr[{\sf Rand}_{\A,G}=1]-(\delta_{{\sf PRG}}+\delta_{{\sf rand}})\right|\\
        &\geq \left|\Pr[{\sf PRG}_{\A,G}=1]-\Pr[{\sf Rand}_{\A,G}=1]\right|-|\delta_{{\sf PRG}}+\delta_{{\sf rand}}| \\
        &= {\sf Adv}_{\A,G}^{\sf PRG}-\left|\delta_{{\sf PRG}}+\delta_{{\sf rand}}\right|.
    \end{align*}
    Lemma~\ref{lemma:transition1} implies $\delta_{{\sf PRG}}\leq \frac{1}{2}{\sf Adv}_{\A,G}^{\sf PRG}$. On the other hand, in ${\sf Rand}_{\B,G}$ and ${\sf Rand}_{\A,G}$ the oracle is independent of $g$ and therefore $\Pr[{\sf Rand}_{\B,G}(g)=1]=\Pr[{\sf Rand}_{\A,G}(g)=0]$. The theorem  follows.
\end{proof}

\section{Simulating Pseudo-Deterministic Algorithms Classically}
We now turn to the case where the PRG is a quantum algorithm which may make quantum queries to the oracle. In this section we show~(1) that any pseudo-deterministic oracle algorithm can only depend on a polynomial number of oracle values, and~(2) that for any pseudo-deterministic quantum oracle algorithm $\A$ making queries to an oracle $H$, there exists a computationally unbounded but query bounded classical algorithm $\B$ that makes only polynomially more (classical) queries to the same oracle $H$, and outputs the same value as $\A$ with high probability. We stress that in this section our results concern fixed (rather than random) oracles. In section 4.3 we show as a corollary that our lifting theorem for PRGs applies to PRGs making quantum queries to the random oracle.

\subsection{From $\delta$-Deterministic to Classical Algorithms}
First we prove that any $\delta$-deterministic oracle algorithm can only depend on a polynomial number of oracle values. 
\begin{lemma}
\label{lemma:oracleDependance}
Let $\A^{(\cdot)}$ be a quantum algorithm making $Q$ quantum queries to an oracle $F:\{0,1\}^{n}\mapsto \{0,1\}^{*}$. Assume that $\A^{(\cdot)}$ is $\delta$-deterministic given oracle access to any function on $n$-bit strings. Suppose that $\A$ makes $Q$ queries to $F$. Then there exists a subset $S \subset \{0,1\}^{n}$ satisfying the following properties:
\begin{enumerate}
    \item $|S|\leq \frac{Q^{4}}{(1-2\delta)^{4}}$
    \item $\A^{\ket{F}}\qeq{\delta} \A^{\ket{H}}$ for any $H$ satisfying $H|_{S}=F|_{S}$
    \item $q_{z}^{F}\geq \frac{(1-2\delta)^{4}}{Q^{3}} \quad \forall z \in S$
\end{enumerate}

\end{lemma}
\begin{proof}
First we define some notation relative to a partial function $R=\{(x_i,y_i)\}_{i=1}^{k}$. Recall that $D_{R}=\{x_{i}\}_{i=1}^{k}$ denotes the domain of $R$. We define the following sequence of functions. For $i=1,...,k$, let $F_{i}=F^{(R\setminus \{(x_{i},y_{i})\})}$. Let $\ket{\phi_{i}}$ be the quantum state that results from the execution of $A^{\ket{F_{i}}}$ (excluding the final measurement). Let $\ket{\phi}$ be the quantum state output by $A^{\ket{F}}$ and let $\ket{\phi_{R}}$ be the quantum state output by $A^{\ket{F^{(R)}}}$. 

Now we describe a procedure to construct the set $S$ from the lemma statement. Let $K=\{0,1\}^{n}$ and let $S=\emptyset$. Let $R$ be the partial function with the smallest domain $D_{R}$, such that $D_{R}\subseteq K$ and $A^{\ket{F}}\not \qeq{\delta}A^{\ket{F^{(R)}}}$. Let $x_{m}$ be the element of $R$ with the largest total query magnitude $q_{x_{m}}^{F}$. Add $x_{m}$ to $S$ and remove $x_{m}$ from $K$. Repeat the above procedure until no such partial function $R$ exists. At this point return $S$.

We now prove that this set $S$ satisfies all properties from the lemma statement. First note that property (2) is just the termination condition for the above algorithm and so $S$ must satisfy this property. We also note that property $(3)$ implies property $(1)$ since the total query magnitude of $A^{\ket{F}}$ across all points is at most $Q$. Thus it remains to prove that property (3) is satisfied.

To show that each $R$ contains some $x_{m}$ such that $q_{x_{m}}^{F}\geq \frac{(1-2\delta)^{4}}{Q^{3}}$, we prove that (1)~each set $R$ is of size at most $\frac{Q^{2}}{(1-2\delta)^{2}}$, and (2)~that  $\sum_{z \in D_{R}}q_{z}^{F}\geq \frac{(1-2\delta)^{2}}{Q}$. These two claims give a lower bound on the average query magnitude across points in $D_{R}$ when executing $\A^{\ket{F}}(x)$, and therefore a lower bound on $q_{x_{m}}^{F}$.

Proceeding with (1), we will upper bound the size of $D_R$ by showing that $\A^{\ket{F^{(R)}}}$ must have queried each point in $D_R$ with a ``large'' weight. Since the total weight of all queries across all points for a $Q$ query algorithm is at most $Q$, this gives an upper bound on the size of $D_R$. Since $R$ is the partial function with the smallest domain such that $\A^{\ket{F}}\not \qeq{\delta} A^{\ket{F^{(R)}}}$, it follows that for any $i \in \{1,...,|D_R|\}$, we have $\A^{\ket{F}}\qeq{\delta}\A^{\ket{F_i}}$ (since $F_{i}$ differs from $F$ on a subset of the domain of size $|D_R|-1<|D_{R}|$), and therefore $\A^{\ket{F_{i}}}\not \qeq{\delta} \A^{\ket{F^{(R)}}}$. Therefore we have the following inequalities, where the first inequality comes from the fact that $\A^{\ket{F_{i}}}\not \qeq{\delta} \A^{\ket{F^{(R)}}}$ together with Lemma~\ref{lemma:measure}, and the second follows from the swapping Lemma:
\begin{equation}
\label{eqn:firstInequality}
(1-2\delta) \leq \norm{\ket{\phi_{R}}-\ket{\phi_{i}}}\leq \sqrt{Q\cdot q_{x_{i}}^{F^{(R)}}} \implies \frac{(1-2\delta)^{2}}{Q}\leq q_{x_{i}}^{F^{(R)}}.    
\end{equation}
In more detail, the first inequality above is due to the following. Since $\A^{\ket{F_i}}\not \qeq{\delta} \A^{\ket{F^{(R)}}}$, it follows that $\ket{\phi_i}$ and $\ket{\phi_{R}}$ can be distinguished with probability at least $1-2\delta$, and therefore the trace distance between $\ket{\phi_i}$ and $\ket{\phi_{R}}$ is at least $(1-2\delta)$. Lemma~\ref{lemma:measure} tells us that the Euclidean distance is lower bounded by the trace distance, and therefore we have $\norm{\ket{\phi_i}-\ket{\phi_{R}}}\geq (1-2\delta)$. Since the total query magnitude of a $Q$ query algorithm (and in particular of $\A^{\ket{F^{(R)}}}$) across all points is at most $Q$, it follows that $|D_R|\leq \frac{Q^{2}}{(1-2\delta)^{2}}$.

Proceeding with (2), note that $\A^{\ket{F}}\not \qeq{\delta}\A^{\ket{F^{(R)}}}$, and therefore
\begin{equation}
    (1-2\delta) \leq \norm{\ket{\phi_{R}}-\ket{\phi}}\leq \sqrt{Q\sum_{z \in D_{R}}q_{z}^{F}}\implies \frac{(1-2\delta)^{2}}{Q}\leq \sum_{z \in D_{R}}q_{z}^{F},
\end{equation}
where the above inequalities are proved in the same way as they were for Equation~\ref{eqn:firstInequality} above.
Combining the above lower bound with the upper bound on $|D_R|$ from before gives
$$
\frac{1}{|D_R|}\sum_{z \in D_{R}}q_{z}^{F}\geq \frac{(1-2\delta)^{4}}{Q^{3}}.
$$
It follows that there is some $x_{m} \in D_{R}$ such that $q_{x_{m}}^{F}\geq \frac{(1-2\delta)^{4}}{Q^{3}}$. 
\end{proof}

\subsection{The Simulation Algorithm}
We now introduce and motivate our classical-query algorithm in Figure~\ref{fig:simOracle1}. Let $\A$ be a $\delta$-deterministic oracle algorithm making $Q_{\A}$ queries to some oracle~$F$. The result in the previous section gives us a set $S_{F}$ such that the output of $\A$ only depends on the oracle values on points in $S_{F}$. Thus, our approach will be to learn $S_{F}$, query the oracle only on values in $S_{F}$, and then classically simulate the computation of $\A$ using an oracle that is consistent with $F$ on $S_{F}$. Since $\A$ only depends on values inside this set, we may define our oracle arbitrarily on all other values. 

The algorithm itself is quite simple. We start with an initially empty partial function $f_{0}$ which we shall modify as we make classical queries to the oracle. For $i$ starting at $0$ we do the following. We classically simulate the execution of $\A^{\ket{\extFunc{f_{i}}}}$ and record all points in the domain of the oracle with query magnitude exceeding $q_{\text{min}}$. We then query $F$ these points and update our partial function to be consistent with the oracle $F$ on these additional points. Let $f_{i+1}$ be the updated partial function. As we prove below, repeating this process a sufficient number of times results in a partial function whose domain contains $S_{F}$, and therefore has all of the required information to compute the most likely output of $\A^{\ket{F}}$.

\noindent\begin{minipage}{\textwidth}
  \centering
     \begin{minipage}{.34\textwidth}
\begin{algorithm}[H]
    \begin{algorithmic}[1]
    \State $S=\left\{x\mid q_{x}^{\extFunc{f}}\geq \floor*{\frac{(1-2\delta)^{4}}{Q^{3}}}\right\}$
    \For{$x \in S$}
    \If{$x \not \in D_{f}$}
    \State $f = f \cup \{(x,F(x))\}$
    \EndIf
    \EndFor 
    \State Return $f$
    \end{algorithmic}
      \caption{${\sf update}^{F}(\A,f)$}
    \label{alg:lazySamp1sub1}
\end{algorithm} 
  \end{minipage}
   \begin{minipage}{.64\textwidth}
\begin{algorithm}[H]
    \begin{algorithmic}[1]
    \State $c = 0$
    \State $y_{\text{old}}=\arg\max_{y}\left\{\Pr[y\leftarrow \A^{\ket{\extFunc{f_0}}}]\right\}$
    \State $k=\ceil*{\frac{Q^{4}}{(1-2\delta)^{4}}}$
    \While{$c\leq k$ and $y_{\text{old}}=\arg\max_{y}\left\{\Pr[y\leftarrow \A^{\ket{\extFunc{f_{c}}}}]\right\}$}
    \State $c\leftarrow c+1$
    \State $f_{c}\leftarrow {\sf update}^{F}(\A,f_{c-1})$
    \EndWhile 
    \State return $(f_c,c)$
    \end{algorithmic}
      \caption{${\sf getPoint}^{F}(\A,f_0)$}
    \label{alg:lazysamp1sub2}
\end{algorithm} 
  \end{minipage}\\
    \begin{minipage}{.5\textwidth}
    \begin{algorithm}[H]
    \begin{algorithmic}[1]
    \State $f_0=\emptyset$
    \State $k=\ceil*{\frac{Q^{4}}{(1-2\delta)^{4}}}$
    \For{$i=1,...,k$}
    \State $(f'_{i-1},c_1)\leftarrow {\sf getPoint}^{F}\left(\A,f_{i-1}\right)$
    \IIf{$c_1 = k$} return $\mathbbm{1}^{(f'_{i-1})}$ \EndIIf 
    \State $(f_i,c_2)\leftarrow {\sf getPoint}^{F}\left(\A,f'_{i-1}\right)$
    \IIf{$c_2 = k$} return $\mathbbm{1}^{(f_{i})}$ \EndIIf 
    \EndFor 
    \State return $\mathbbm{1}^{(f_{k})}$.
    \end{algorithmic}
      \caption{${\sf simOracle}^{F}(\A)$}
    \label{alg:lazySamp1}
\end{algorithm} 
  \end{minipage}
  \captionof{figure}{Classical query algorithm for simulating quantum oracle algorithms. \vspace*{8pt}}
  \label{fig:simOracle1}
\end{minipage}

We now give a brief overview of the algorithms  in Figure 4, along with a sketch of the analysis. The proof of correctness for the main simulation algorithm ${\sf simOracle}$ (Algorithm~\ref{alg:lazySamp1}) is given in Lemma~\ref{lemma:OracleSimulation}, and an analysis of the subroutine ${\sf getPoint}$ (Algorithm~\ref{alg:lazysamp1sub2}) is given in Lemma~\ref{lemma:addingPoints1}.

\begin{itemize}
    \item ${\sf update}^{F}(\A,f)$: Given classical oracle access to a function $F:\{0,1\}^{n}\mapsto \{0,1\}^{*}$, and on input a description of a quantum algorithm $\A$, and a partial function $f$ such that $F\in {\sf Func}_{n,m}(f)$, ${\sf update}$ first classically simulates the computation $\A^{\ket{\extFunc{f}}}$, and finds all $x \in \bool^{n}$ such that $q_{x}^{\extFunc{f}}\geq \floor*{\frac{(1-2\delta)^{4}}{Q^{3}}}$. The oracle $F$ is then queried on each of these points, and the partial function $f$ is extended so that $f(x)=F(x)$ for each point $x$ on which $F$ was queried. 
    \item ${\sf getPoint}^{F}(\A,f_0)$: Given classical oracle access to a function $F:\{0,1\}^{n}\mapsto \bool^{*}$, and on input a description of a quantum algorithm $\A$, and a partial function $f_0$ such that $F\in {\sf Func}_{n,m}(f_0)$, ${\sf getPoint}$ first computes the value $y_0 = \arg\max_{x}\left\{\Pr[y\leftarrow \A^{\ket{\extFunc{f_0}}}]\right\}$. Then, starting at $i=0$, ${\sf getPoint}$ updates the partial function by computing $f_{i+1}\leftarrow {\sf update}^{F}(\A,f_i)$. This process terminates either when $\ceil*{\frac{Q^{4}}{(1-2\delta)^{4}}}$ iterations of this update process have passed, or when for some iteration $i$, the quantity $\arg\max_{y}\left\{\Pr[y\leftarrow \A^{\ket{\extFunc{f_{i}}}}]\right\}$ no longer equals the initial most likely value $y_0$. In Lemma~\ref{lemma:addingPoints1}, we will prove several facts about ${\sf getPoint}$, the implications of which we discuss here. Starting with some partial function $f_{i}$, consider the following sequence of two invocations of ${\sf getPoint}$:
$$
(f_{i}',c_{1})\leftarrow {\sf getPoint}^{F}(\A,f_{i}), \quad \quad (f_{i+1},c_{2})\leftarrow {\sf getPoint}^{F}(\A,f_{i}').
$$
It turns out that if $c_{1}=\ceil*{\frac{Q^{4}}{(1-2\delta)^{4}}}$, then $\A^{\ket{\extFunc{f_{i}'}}}\qeq{\delta}\A^{\ket{F}}$, and similarly if $c_2=\ceil*{\frac{Q^{4}}{(1-2\delta)^{4}}}$, then $\A^{\ket{\extFunc{f_{i+1}}}}\qeq{\delta}\A^{\ket{F}}$. On the other hand, if $c_1,c_2<\ceil*{\frac{Q^{4}}{(1-2\delta)^{4}}}$, then there is some $x \in S_{F}$ which is in the domain of $f_{i+1}$ but not in the domain of $f_{i}$. In other words, invoking ${\sf getPoint}$ twice either returns a partial function that can be used to simulate $\A^{\ket{F}}$, or it results in learning a new element of $S_{F}$. Therefore repeating the above sequence of two invocations $k=\ceil*{\frac{Q^{4}}{(1-2\delta)^{4}}}$ times should result in the domain of $f_{k}$ containing all of $S_{F}$ as $|S_{F}|\leq k$. With this in mind, the main algorithm is as follows.
    \item ${\sf simOracle}^{F}(\A)$: Given classical oracle access to a function $F:\{0,1\}^{n}\mapsto \bool^{*}$, and on input a description of a quantum algorithm $\A$, ${\sf simOracle}$ carries out the following process up to $k=\ceil*{\frac{Q^{4}}{(1-2\delta)^{4}}}$ times. Compute two invocations of ${\sf getPoint}$ as described above. If either of $c_1$ or $c_2$ are equal to $\ceil*{\frac{Q^{4}}{(1-2\delta)^{4}}}$, then return $\extFunc{f}$, where $f$ is the corresponding partial function that was returned by ${\sf getPoint}$. Otherwise, return $\extFunc{f_{k}}$ after $k$ iterations of this process. The proof of correctness is based on the discussion in the previous paragraph. We now proceed to the formal proofs.
\end{itemize}

\begin{lemma}
\label{lemma:addingPoints1}
Let $\A$ be a $\delta$-deterministic quantum algorithm given quantum oracle access to a function $F$ and making $Q$ queries. Let $S_{F}\subseteq D_{F}$ be such that $A^{\ket{F}}\qeq{\delta}A^{\ket{H}}$ for any $H$ such that $H|_{S_{F}}=F|_{S_{F}}$. Let $f_{0}$ be some partial function such that $F \in {\sf Func}_{n,m}(f_0)$. Suppose that we run the function (described in Algorithm~\ref{alg:lazysamp1sub2})
$$
(f_{c},c)\leftarrow {\sf getPoint}^{F}\left(\A,f_{0}\right)
$$ 
and let $(f_{c},c)$ be the output. Then:
\begin{enumerate}
    \item If $\A^{\ket{F}}\not\qeq{\delta}\A^{\ket{\mathbbm{1}^{(f_{0})}}}$, then there exists some $x \in S_{F}$ such that $x$ is in the domain of $f_{c}$ but not in the domain of $f_{0}$.
    \item If $\A^{\ket{F}}\not\qeq{\delta}\A^{\ket{\mathbbm{1}^{(f_{0})}}}$, then there exists some $j<\ceil*{\frac{Q^{4}}{(1-2\delta)^{4}}}$ such that $\A^{\ket{\mathbbm{1}^{(f_{j})}}}\not \qeq{\delta}\A^{\ket{\mathbbm{1}^{(f_{0})}}}$, where $f_{j}$ refers to the state of the partial function after the $j$'th iteration of the for loop in Algorithm~\ref{alg:lazysamp1sub2}.
    \item If $c=\ceil*{\frac{Q^{4}}{(1-2\delta)^{4}}}$, then $A^{\ket{\mathbbm{1}^{(f_{0})}}}\qeq{\delta}A^{\ket{F}}$.
\end{enumerate}
\end{lemma}
\begin{proof}
For notational compactness in the following calculations we let $k=\ceil*{\frac{Q^{4}}{(1-2\delta)^{4}}}$. As in the lemma statement let $i$ refer to some iteration of the while loop of Algorithm~\ref{alg:lazysamp1sub2}. We first note that item~(3) is implied by~(2). We now prove item~(1). To start, we show that for any $i \in \{0,...,k-1\}$, if $\A^{\ket{\mathbbm{1}^{(f_{i})}}}\not \qeq{\delta}\A^{\ket{F}}$, then there exists some $x \in S_{F}$ in the domain of $f_{i+1}$ which is not in the domain of $f_{i}$. In other words, we show that after another iteration of the while loop, a new element of $S_{F}$ has been learned. Assuming that $\A^{\ket{F}}\not\qeq{\delta}\A^{\ket{\extFunc{f_{i}}}}$, we first show there must exist some $x \in S_{F}\cap S_{\extFunc{f_{i}}}$ such that $F(x)\neq \extFunc{f_{i}}(x)$. If there were no such $x$, then we could define the function
\[ F'(x) = \begin{cases} 
      F(x) & x \in S_{F} \\
      \extFunc{f_{i}}(x) & \text{otherwise}. 
   \end{cases}
\]
Applying Lemma~\ref{lemma:oracleDependance} we have $A^{\ket{F'}}\qeq{\delta}\A^{\ket{F}}$ since $F'|_{S_{F}}=F|_{S_{F}}$, and $\A^{\ket{F'}}\qeq{\delta}\A^{\ket{\extFunc{f_{i}}}}$ since $\extFunc{f_{i}}|_{S_{\extFunc{f_{i}}}}=F'|_{S_{\extFunc{f_{i}}}}$, which contradicts the assumption that $\A^{\ket{F}}\not \qeq{\delta}\A^{\ket{\extFunc{f_{i}}}}$. However since $x \in S_{\extFunc{f_{i}}}$, Lemma~\ref{lemma:oracleDependance} tells us that the query magnitude of $\A^{\ket{\extFunc{f_{i}}}}$ on $x$ will exceed $\floor*{\frac{(1-2\delta)^{4}}{Q^{3}}}$, and therefore the algorithm ${\sf update}(\A,f_{i})$ will query $F$ on $x$. It follows that the domain of $f_{i+1}$ will contain $x$. This proves item~(1) since if this element $x$ is in the domain of $f_{i+1}$ it is also in the domain of $f_{c}$ since $c\geq i+1$.

To prove item~(2), suppose that $\A^{\ket{\extFunc{f_{0}}}}\qeq{\delta}...\qeq{\delta}\A^{\ket{\extFunc{f_{k-1}}}}$. It follows from the above that the domain of $f_{k-1}$ contains $k-1$ values in $S_{F}$. We also have that $\A^{\ket{\extFunc{f_{k-1}}}}\not \qeq{\delta}\A^{\ket{F}}$ since one of the assumptions in the lemma statement is that $\A^{\ket{\extFunc{f_{0}}}}\not \qeq{\delta}\A^{\ket{F}}$. It follows from the above that the domain of $f_{k}$ contains $k$ elements in $S_{F}$, but since $|S_{F}|\leq k$, the domain of $f_{k}$ contains all of $S_{F}$. This, together with Lemma $\ref{lemma:oracleDependance}$ implies that $\A^{\ket{\extFunc{f_{k}}}}\qeq{\delta}\A^{\ket{F}}$ which proves item~(2).   
\end{proof}

We are now ready to prove the main result below. 
\begin{lemma}[Oracle simulation]
\label{lemma:OracleSimulation}
Let $\A$ be a $\delta$-deterministic quantum algorithm making $Q$ queries to an oracle $F:\{0,1\}^{n}\mapsto \{0,1\}^{*}$. Then Algorithm~\ref{alg:lazySamp1} makes at most $\frac{2Q^{12}}{(1-2\delta)^{12}}$ classical queries to $F$ and returns a function $H:\{0,1\}^{n}\mapsto \{0,1\}^{*}$ such that $\A^{\ket{F}}\qeq{\delta}\A^{\ket{H}}$.
\end{lemma}
\begin{proof}
Each execution of ${\sf getPoint}$ calls ${\sf update}$ at most $\ceil*{\frac{Q^{4}}{(1-2\delta)^{4}}}$ times and ${\sf getPoint}$ is called at most $2\ceil*{\frac{Q^{4}}{(1-2\delta)^{4}}}$ times. Each execution of ${\sf update}$ queries $F$ at most $\ceil*{\frac{Q^{4}}{(1-2\delta)^{4}}}$ times which gives the stated query complexity.

Now we show correctness of ${\sf simOracle}$ (Algorithm~\ref{alg:lazySamp1}). We consider two ways that Algorithm~\ref{alg:lazySamp1} can terminate. The first case, is that the algorithm terminates on lines~5 or~7 (we consider these two as a single case), and the second case is that the algorithm terminates on line~9. We show that in either case, the function $\extFunc{f}$ that is returned satisfies $\A^{\ket{\extFunc{f}}}\qeq{\delta}\A^{\ket{F}}$. Proceeding with the first case, suppose that $c_{b}=\ceil{\frac{Q^{4}}{(1-2\delta)^{4}}}$ for some $b \in \{0,1\}$ as is required to terminate on lines 5 or 7. It follows from Lemma~\ref{lemma:addingPoints1} (item~3) that in this case the function $f_{c}$ output by ${\sf getPoint}$ was such that~$\A^{\ket{\extFunc{f_{c}}}}\qeq{\delta}\A^{\ket{F}}$.

Now we consider the case that the algorithm terminates on line~9. In keeping with the notation used in Algorithm~\ref{alg:lazySamp1}, for the $i$'th iteration of the for loop, let $f_{i}$ be the partial function prior to the first invocation of ${\sf getPoint}$, let $f'_{i}$ be the partial function returned by the first invocation of ${\sf getPoint}$, and let $f_{i+1}$ be the partial function returned by the second invocation of ${\sf getPoint}$. We will prove that there is some $x \in S_{F}$ such that the domain of $f_{i+1}$ contains $x$ but the domain of $f_{i}$ does not contain $x$. It will then follow that after $k=\ceil*{\frac{Q^{4}}{(1-2\delta)^{4}}}$ iterations of the for loop, the domain of $f_{k}$ will contain all of $S_{F}$ since $|S_{F}|\leq k$. Therefore by Lemma~\ref{lemma:oracleDependance}, it follows that $\A^{\ket{\extFunc{f_{k}}}}\qeq{\delta}\A^{\ket{F}}$.

First suppose that $\A^{\ket{\mathbbm{1}^{(f_{i})}}}\not \qeq{\delta}\A^{\ket{F}}$. Then it follows from Lemma~\ref{lemma:addingPoints1} (item~3) that the domain of $f'_{i}$ contains an element of $S_{F}$ which is not in the domain of~$f_{i}$. It follows that $f_{i+1}$ also contains this point. Now suppose on the other hand that $\A^{\ket{\mathbbm{1}^{(f_{i})}}}\qeq{\delta}\A^{\ket{F}}$. Since by assumption $c<\ceil*{\frac{Q^{4}}{(1-2\delta)^{4}}}$, the subroutine ${\sf getPoint}$ must have terminated because $\A^{\ket{\mathbbm{1}^{(f_{i})}}}\not \qeq{\delta}\A^{\ket{\mathbbm{1}^{(f'_{i})}}}$. It follows that $\A^{\ket{\mathbbm{1}^{(f'_{i})}}}\not \qeq{\delta}\A^{\ket{F}}$. It follows once again from Lemma~\ref{lemma:addingPoints1} (item~3) that the second invocation of ${\sf getPoint}$ will result in the domain of $f_{i+1}$ containing some element of $S_{F}$ which is not contained in $f'_{i}$, and therefore was not contained in~$f_{i}$. 
\end{proof}

\subsection{Lifting Theorem for PRGs Making Quantum Queries}
We can apply the above results to show that our PRG lifting theorem applies even if the PRG itself makes quantum queries to the random oracle. For a $\delta$-deterministic PRG $G_{1}^{\ket{H}}$ making quantum queries to its oracle, we can apply Lemma~\ref{lemma:OracleSimulation} to construct a PRG $G_{2}^{H}$ making a polynomially related number of classical queries to the same oracle and having (almost) the same output as~$G_{1}$. It is a straightforward exercise (Lemma~\ref{lemma:classical_quantum_prg_advantage}) to show that for any quantum distinguisher $\mathcal{A}_{1}$ against $G_{1}$, there is a quantum distinguisher $\mathcal{A}_{2}$ against $G_{2}$ with almost the same advantage. As the latter distinguisher is interacting with a PRG making classical queries, our lifting theorem applies and we can construct a classical distinguisher $\B$ against $G_{2}$ with advantage almost as large as that of~$\mathcal{A}_{2}$. 
\begin{lemma}
\label{lemma:classical_quantum_prg_advantage}
    Let $G_{1}^{\ket{H}}$ be a $\delta$-deterministic PRG making quantum queries to a random oracle. Let $G_{2}^{H}$ be a PRG making classical queries to a random oracle such that $\forall s$,
    $$
        G_{1}^{\ket{H}}(s)\qeq{\delta} G_{2}^{H}(s). 
    $$
    Let $A_{1}^{\ket{H}}$ be a distinguisher with quantum oracle access to $H$ with distinguishing advantage against $G_{1}$ given by ${\sf Adv}_{\ket{\mathcal{A}_{1}},G_1}^{\sf PRG}$. Then there exists $\mathcal{A}_{2}^{\ket{H}}$ such that 
    $$
        {\sf Adv}_{\A_2,G_{2}}^{\sf PRG}\geq {\sf Adv}_{\A_{1},G_{1}}^{\sf PRG}-\delta.
    $$
\end{lemma}
\begin{proof}
    On input $g$, $\A_{2}(g)$ simply runs $\A_{1}(g)$ and returns the resulting output. First note that 
    $$
    \Pr[{\sf Rand}_{\A_{1},G_{1}}=1]=\Pr[{\sf Rand}_{\A_{2},G_2}=1]
    $$ 
    since $H$ and $g$ are distributed identically in each case. In the case of ${\sf PRG}_{\A_{1},G_1}$ and ${\sf PRG}_{\A_{2},G_2}$, $H$ is distributed identically in each case, but $g$ may not be due to the fact that $G_1$ is only $\delta$-deterministic. It follows from the definition that the two distributions on $g$ are at most $\delta$ apart, and therefore the distributions on the outputs of $\A_1$ in each case are also only a distance of $\delta$ apart. Therefore 
    $$
        \left|\Pr[{\sf PRG}_{\A_{1},G_1}=1]-\Pr[{\sf PRG}_{\A_{2},G_2}=1]\right|\leq \delta.
    $$
    The lemma follows.
\end{proof}

\begin{lemma}
Let $G^{\ket{H}}:\bool^{k}\mapsto \bool^{\ell}$ be a $\delta$-deterministic PRG making $Q_{G}$ quantum queries to a random oracle $H:\bool^{n}\mapsto \bool^{m}$, and let $\A$ be a quantum distinguisher making $Q_{\A}$ quantum queries to the same oracle. Let $\epsilon = {\sf Adv}_{\A,G}^{\sf PRG}$. Then there exists a classical query algorithm $\B^{H}$ making ${\sf poly}(k,n,Q_{\A},Q_{G},\epsilon^{-1}$) classical queries to the random oracle such that
    $$
        {\sf Adv}^{{\sf PRG}}_{\B,G}\geq \frac{\epsilon}{2}-\delta.
    $$    
\end{lemma}
\begin{proof}
    Let $G_2^{H}$ be a classical-query algorithm that on input $s$ outputs $\arg\max_{y}\{\Pr[y\leftarrow G^{\ket{H}}(s)]\}$ and makes $Q_{G_{2}}={\sf poly}(Q_{G})$ queries to $H$. Such an algorithm exists by Lemma~\ref{lemma:OracleSimulation}. It follows from Lemma~\ref{lemma:classical_quantum_prg_advantage} that there exists a quantum distinguisher $\A_2^{\ket{H}}$ such that 
    $$
        {\sf Adv}_{\A_{2},G_{2}}^{\sf PRG}\geq {\sf Adv}_{\A,G}^{\sf PRG}-\delta.
    $$
    Applying our lifting theorem for classical PRGs to the quantum distinguisher $\A_2$, it follows that there exists some classical distinguisher $\B^{H}$ making ${\sf poly}(k,n,Q_{\A_2},Q_{G},\epsilon^{-1})$ such that 
    $$
        {\sf Adv}_{\B,G_2}^{\sf PRG}\geq \frac{1}{2}{\sf Adv}_{\A_{2},G_2}^{\sf PRG}.
    $$
    Putting the above two inequalities together we obtain the lemma.
\end{proof}

 \bibliographystyle{plain}
\bibliography{main}

\end{document}